\documentclass[graybox]{svmult}
\usepackage{newtxtext}       %
\usepackage{newtxmath}
\usepackage{scrextend}
\usepackage{hyperref,lineno}
\usepackage{graphicx}
\graphicspath{ {image/} }
\usepackage{subcaption}
\usepackage{amsmath,eqnarray}
\usepackage{multirow}
\usepackage[linesnumbered,ruled]{algorithm2e}
\newcommand{\mybegineq}{\begin{linenomath}\begin{equation}}
\newcommand{\myendeq}{\end{equation}\end{linenomath}}
\newtheorem{lem}{Lemma}
\newtheorem{thm}{Theorem}
\newtheorem{cor}{Corollary}
\usepackage[bottom]{footmisc}

\begin{document}
\title*{Sizing up the Batteries: Modelling of Energy-Harvesting Sensor Nodes in a Delay Tolerant Network}
\titlerunning{Sizing Up The Batteries}
\author{Jeremiah D. Deng}
\institute{Jeremiah~D.~Deng
\at Department of Information Science, University of Otago, PO Box 56, Dunedin 9054, 
\email{jeremiah.deng@otago.ac.nz}}
\date{}
\maketitle

\abstract{
For energy-harvesting sensor nodes, rechargeable batteries play a critical role in sensing and transmissions. By coupling two simple Markovian queue models in a delay-tolerant networking setting, we consider the problem of battery sizing for these sensor nodes to operate effectively: given the intended energy depletion and overflow probabilities, how to decide the minimal battery capacity that is required to ensure opportunistic data exchange despite the inherent intermittency of renewable energy generation. }


\section{Introduction}
Recently, energy-harvesting wireless sensor networks (EH-WSN)~\cite{Adu-Manu:18} have become a promising technology for sensing applications. The advantage of EH-WSN is obvious - batteries on the sensor nodes can be downsized due to their energy-harvesting capability, the network enjoys longer life time, eliminating the need of frequent of battery replacement, which is especially challenging for large-scale sensor deployment. However, apart from reservoirs, most renewable energy sources are intermittent in nature, which raises new challenges in designing EH-WSNs. For example, sensors may not get proper sunshine for recharging for hours, and wearable devices operated by kinetic energy will not benefit much from humans sitting for hours. This implies the necessity of using batteries to buffer the unsteady power supply from renewable energy sources. 

We consider a generic EH-WSN scenario where mobile nodes are equipped with capacity-limited batteries that are powered by harvested kinetic energy; data exchange between nodes requires 1) they are within transmission range to each other; and 2)  there is sufficient energy to conduct data transmission. This is in effect an EH-WSN operating as a delay-tolerant network (DTN)~\cite{Rodrigues:14}, where data transmission is opportunistic. In such a scenario, it is both important to ensure the battery size is large enough to avoid energy depletion (and hence potential failure for transmission) and energy overflow, both detrimental to the battery life. 

In a previous work~\cite{Zareei:18}, we have examined battery sizing in terms of depletion probability and overflow probability respectively, using a coupled data and energy queue system. In this work, we intend to investigate the mathematical properties of battery size as a function regarding the operational probability requirement, and develop an algorithm to calculate the minimum battery size needed to meet the given requirements.    

\section{Related Work}
As a performance modelling tool, queueing theory has been employed to
study EH-WSNs. Gelenbe~\cite{gelenbe2015} first looked the modelling of an EH-sensor node using the concept of discretized energy unit called ``energy packets''. The arrival of these energy packets is assumed to follow a Poisson process. A routing approach was further developed in~\cite{gelenbe2015interconnected}. 
A more general queueing model was introduced in~\cite{kadioglupacket}, relaxing the assumption that exactly one energy packet is required to transmit a data packet. A Markovian model with data buffering was further considered in~\cite{DeCuypere:18}. In a recent work~\cite{Zareei:18} we showed that kinetic energy harvested by fitness gears discretized as energy packets can be well modelled by Poisson processes. These previous works, however, considered only static EH sensors, without involving potential intermittent connections between EH-sensor nodes due to mobility. 

On the other hand, mobility has been widely investigated in ordinary wireless sensor networks and DTNs~\cite{krifa2008optimal,patel2015}. Despite some counter-arguments~\cite{Chaintreau:07}, several mobility model studies~\cite{aldous2002reversible,Groenevelt:2005,spyropoulos2006,krifa2008optimal} suggested that two mobile nodes' encounter follows a Poisson process in mobile ad hoc networks and DTNs.
There are few studies on energy harvesting networks that investigated the effects
of intermittent connections~\cite{harvestDTN1, lu2014}. 


\section{System Modelling}
Notations used in this article are listed as follows: 
\begin{labeling}{alligator}
\item [$ \lambda_E $] energy packet arrival rate 
\item[$\lambda_D$] data packet arrival rate
\item[$\lambda_C$] connection arrival rate  
\item[$\gamma_D$] ratio $\lambda_D/\lambda_C$
\item[$\gamma_E$] ratio $\lambda_E/\lambda_C$
\item[$\gamma$] ratio $\lambda_D/\lambda_E$
\item[$P_{D_0}$] proportion of time that there is no data in the system
\item[$P_{E_k}$] proportion of time that system have $ k $ energy packets $k=0,...,K$
\item[$\rho_D$] utilization factor of data buffer 
\item[$\rho_E$] utilization factor of energy buffer
\item[$\alpha$] acceptable probability of energy depletion
\item[$\beta$] acceptable probability of energy overflow  
\item[$K_{\alpha}$] battery capacity decided based on $\alpha$ 
\item[$ K_{\beta}$] battery capacity decided based on $ \beta $ 
\item[$\lceil x \rceil$] ceiling, the greatest integer more than or equal to $ x $  
\end{labeling}

\subsection{The queueing model} \label{sys_char}
We consider a network of mobile EH-sensors. Energy harvesting leads to Poisson arrivals
of energy packets (EP) with a rate of $\lambda_E$. Energy consumption occurs when there are 
data packets in buffer, provided that there are nodes in proximity, which is modulated by another Poisson arrival rate $\lambda_C$. Thus an Energy queue is formed at each sensor node, which can be modelled as an $M/M/1/K$, where $K$ is the battery capacity (in terms of number of energy packets). Data packets (DP) arrive at a Poisson rate $\lambda_D$, and leave a node if there is a connection available and there is at least an energy packet in system.
As memory in a sensor node is relatively cheap and less constrained, for simplicity we set no limit to the data buffer, hence allowing the data queue to be modelled by an $M/M/1$.
Clearly both queues are coupled by the connection availability. Hence the Markovian packet departures in both the Energy queue and the Data queue are modulated by the connection arrival rate $\lambda_C$. 

The system diagram for a sensor node is shown in Figure~\ref{fig:sys}.
\begin{figure}[tbhp]
\centerline{\includegraphics[width=0.4\textwidth]{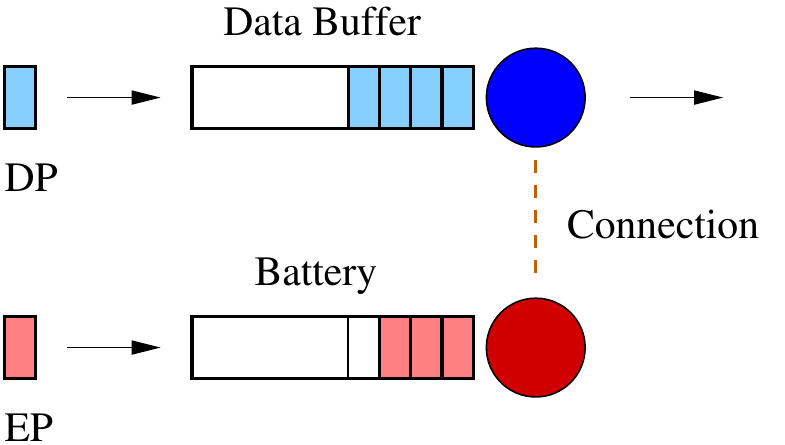}}
\caption{System diagram of a mobile sensor node. DPs arrive in the data buffer, while EPs arrive in the battery. Consumption of the energy as well as the transmission of data occur simultaneously when triggered by a connection established with another node.}
\label{fig:sys}
\end{figure}

Similar to~\cite{kadioglupacket} and~\cite{gelenbe2015interconnected}, we focus on modelling energy needed for data transmission and  
assume that compared with data transmission the sensing process consumes insignificant amount of energy from the battery. 
    

It is also worth mentioning that data transmission time is much faster than energy harvesting in a node. Given the size of sensory DP and the relatively large bandwidth in a network,  
packet transmission time is negligible~\cite{gelenbe2015}.    
Finally, to simplify our analysis we assume at each encounter only one DP is transmitted. 

These assumptions allow us to have a tractable system model with the energy and data state diagrams shown in Figure~\ref{fig:M-M-1-S}. 

\subsection{Queueing analysis} \label{math_analyse}

Queueing analysis has been carried out in the previous work~\cite{Zareei:18}. Here we only summarize some main results. 

The utilization of the data queue is given by:
\mybegineq
\label{rho_D}
\rho_D=\frac{\lambda_D}{\lambda_C(1-P_{E_0})}.
\myendeq
For sake of system stability, we have $\rho_D<1$. According to queueing theory, we have
\mybegineq
\label{PD0}
P_{D_0}=1-\rho_D.
\myendeq

The utilization of the energy queue is
\mybegineq
\label{rho_E}
\rho_E=\frac{\lambda_E}{\lambda_C(1-P_{D_0})}.
\myendeq

And the probability of energy depletion is
\mybegineq
\label{eq:P0}
P_{E_0}=\frac{1-{\rho_E}}{1-{\rho_E}^{K+1}},
\myendeq
while the probability of energy overflow is
\mybegineq
\label{eq:PK}
P_{E_K}={\rho_E}^K P_{E_0},
\myendeq



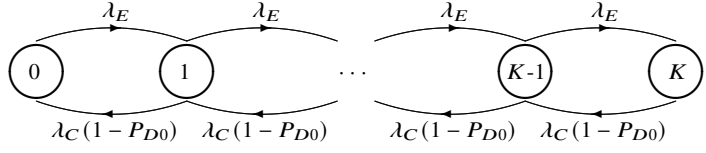
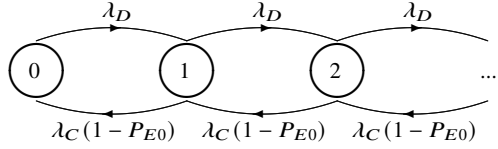
\begin{figure}[!t]
\begin{subfigure}[t]{\textwidth}
\centering
\setlength{\unitlength}{0.25cm}
\begin{picture}(40,10)

\thicklines
\put(1,4){\circle{3}}
\put(9,4){\circle{3}}
\put(27,4){\circle{3}}
\put(35,4){\circle{3}}
\thinlines
\footnotesize
\qbezier(1,5.6)(5,7)(9,5.6)\put(5,6.3){\vector(1,0){0.5}}
\put(0.6,3.6){{0}}
\put(4.5,6.85){{$\lambda_E$}}
\qbezier(1,2.4)(5,1.0)(9,2.4)\put(5,1.7){\vector(-1,0){0.5}}
\put(1.8,0.5){{$\lambda_C(1-P_{D0})$}}
\put(8.6,3.6){{1}}
\qbezier(9,5.6)(13,7)(17,5.6)\put(13,6.3){\vector(1,0){0.5}}
\put(12.5,6.85){{$\lambda_E$}}
\qbezier(9,2.4)(13,1)(17,2.4)\put(13,1.7){\vector(-1,0){0.5}}
\put(9.8,0.5){{$\lambda_C(1-P_{D0})$}}
%
\put(17,3.6){{$\cdots$}}
\qbezier(19,5.6)(23,7)(27,5.6)\put(23,6.3){\vector(1,0){0.5}}
\put(22.5,6.85){{$\lambda_E$}}
\qbezier(19,2.4)(23,1)(27,2.4)\put(23,1.7){\vector(-1,0){0.5}}
\put(19.8,0.5){{$\lambda_C(1-P_{D0})$}}
\put(26.0,3.6){{$K$-1}}
\qbezier(27,5.6)(31,7)(35,5.6)\put(31,6.3){\vector(1,0){0.5}}
\put(30.5,6.85){{$\lambda_E$}}
\qbezier(27,2.4)(31,1)(35,2.4)\put(31,1.7){\vector(-1,0){0.5}}
\put(27.8,0.5){{$\lambda_C(1-P_{D0})$}}
\put(34.6,3.6){{$K$}}

\end{picture}
\caption{The Energy queue modelled as M/M/1/K.}
\end{subfigure}

\begin{subfigure}[t]{\textwidth}
\centering
\setlength{\unitlength}{0.25cm}
\begin{picture}(30,10)
\thicklines
\put(1,4){\circle{3}}
\put(9,4){\circle{3}}
\put(17,4){\circle{3}}
\thinlines
\footnotesize
\qbezier(1,5.6)(5,7)(9,5.6)\put(5,6.3){\vector(1,0){0.5}}
\put(0.6,3.6){{0}}
\put(4.5,6.85){{$\lambda_D$}}
\qbezier(1,2.4)(5,1)(9,2.4)\put(5,1.7){\vector(-1,0){0.5}}
\put(1.8,0.5){{$\lambda_C(1-P_{E0})$}}
\put(8.6,3.6){{1}}
\qbezier(9,5.6)(13,7)(17,5.6)\put(13,6.3){\vector(1,0){0.5}}
\put(12.5,6.85){{$\lambda_D$}}
\qbezier(9,2.4)(13,1)(17,2.4)\put(13,1.7){\vector(-1,0){0.5}}
\put(9.8,0.5){{$\lambda_C(1-P_{E0})$}}
\put(16.6,3.6){{2}}
\qbezier(17,5.6)(21,7)(25,5.6)\put(21,6.3){\vector(1,0){0.5}}
\put(20.5,6.85){{$\lambda_D$}}
\qbezier(17,2.4)(21,1)(25,2.4)\put(21,1.7){\vector(-1,0){0.5}}
\put(17.8,0.5){{$\lambda_C(1-P_{E0})$}}
%
\put(24.6,3.6){{...}}
%

\end{picture}
\caption{The Data queue modelled as M/M/1.}
\end{subfigure}
\caption{Queueing models for the Energy and Data queues respectively. }
\label{fig:M-M-1-S}
\end{figure}
      
By substituting Eq.(\ref{rho_D}) and Eq.(\ref{rho_E}) in Eq.(\ref{PD0}), we have
\mybegineq
\label{P-d-0}
P_{D_0}=1-\frac{\lambda_D}{\lambda_C(1-P_{E_0})}.
\myendeq
Similarly, from (\ref{eq:P0}), we have
%
\mybegineq
\label{PE0}
P_{E_0}=\frac{\left(\lambda_C(1-P_{D_0})\right)^K\left(\lambda_C(1-P_{D_0})-\lambda_E\right)}{\left(\lambda_C(1-P_{D_0})\right)^{K+1}-{\lambda_E}^{K+1}},
\myendeq
which, by substituting $P_{D_0}$ using Eq.(\ref{P-d-0}), becomes
\mybegineq
\begin{array}{ll}
P_{E_0} & =\displaystyle\frac{\lambda_D^K(\lambda_D-\lambda_E(1-P_{E_0}))}{\lambda_D^{K+1}-(\lambda_E(1-P_{E_0}))^{K+1}} \\
& =\displaystyle\frac{\gamma^K(\gamma+P_{E_0}-1)}{\gamma^{K+1}-(1-P_{E_0})^{K+1}}. 
\end{array}
\label{P_0}
\myendeq
where $\gamma=\lambda_D/\lambda_E$. 
Here we introduce a new variable $\zeta=\frac{1-P_{E_0}}{\gamma}$, to further simplify the mathematical formulations. From the equation above, we have:
\mybegineq
1-\gamma\zeta=\displaystyle\frac{1-\zeta}{1-\zeta^{K+1}}.
\label{eq:P_0a}
\myendeq

Using Eq.(\ref{eq:PK}), we work on the probability of energy overflow 
\mybegineq
P_{E_K}=\displaystyle\frac{(1-P_{E_0})^K(\gamma+P_{E_0}-1)}{\gamma^{K+1}-(1-P_{E_0})^{K+1}},
\myendeq
which can be further simplified to 
\mybegineq
\label{P_K}
P_{E_K}=
\begin{cases}
\displaystyle\frac{1 - \frac{1}{\zeta}}{1 - \frac{1}{\zeta^{K+1}}}
& \zeta \ne 1, \\
\frac{1}{K+1} & \zeta=1.
\end{cases}
\myendeq
where the case of $\zeta=1$ is obtained by using L'H\^{o}pital's rule. 

\begin{figure}[!t]
	\centering
	\includegraphics[width=\textwidth]{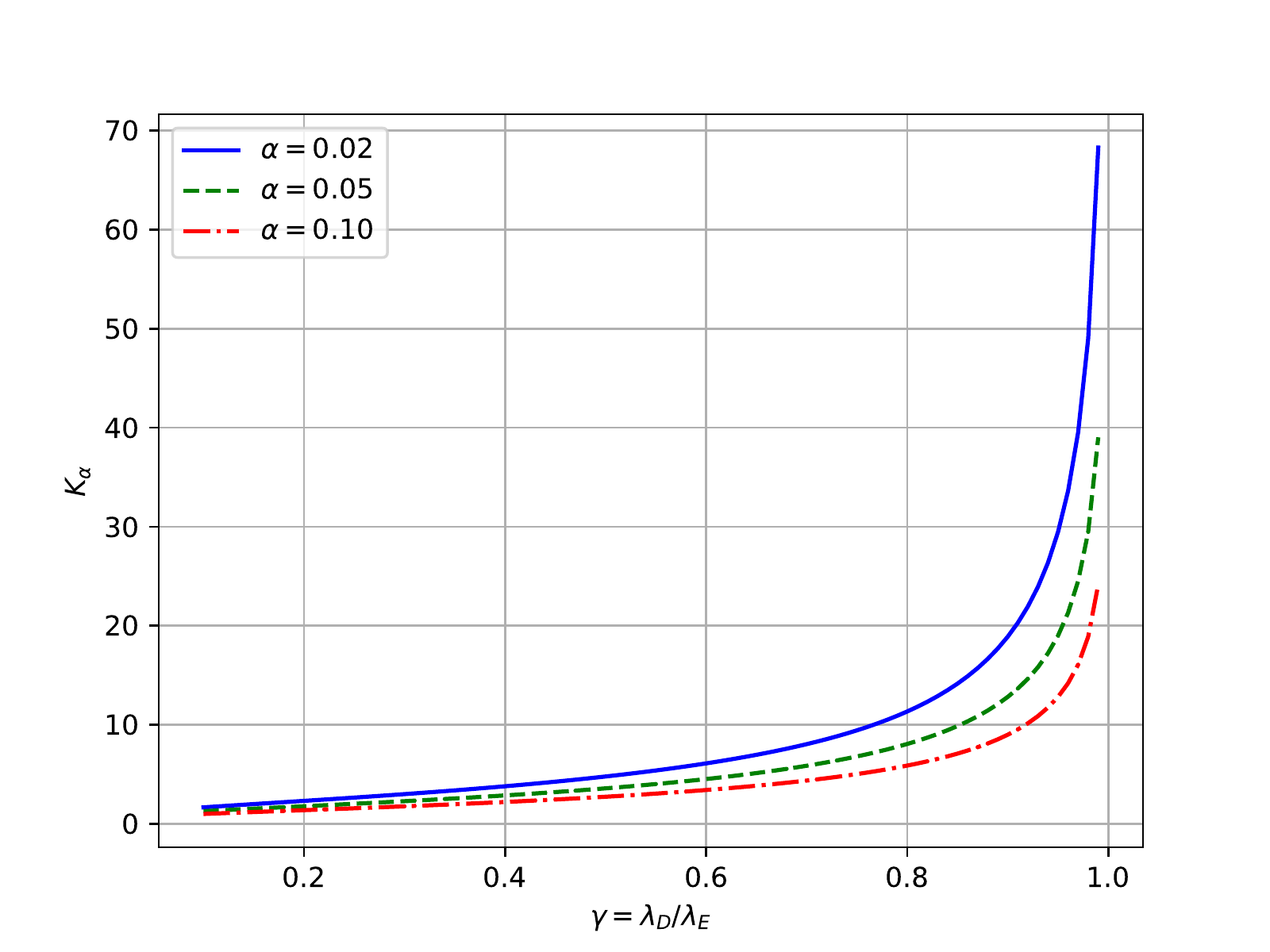}
	\caption{Different values of $K_\alpha$ with respect to $\gamma$ for $\alpha = 0.05 , 0.02$, and $0.1$.} 
	\label{fig:K_a}
\end{figure}


\section{Battery capacity sizing}
Having obtained the formulae for $P_{E_0}$ and $P_{E_K}$, we are now set to find the close-form solution for the battery size as 
required by battery depletion and overflow probabilities. To simplify notations, let $P_{E_0}=\alpha$, $P_{E_K}=\beta$. 
\subsection{Battery size versus depletion probability}
First we look at the battery size decided by $\alpha$, denoted by $K_\alpha$. 
$K_\alpha$ is in fact a function of $\alpha$ and $\gamma$. 
From Eq.(\ref{eq:P_0a}), we have 
\mybegineq
1-\gamma\zeta=\frac{1-\zeta}{1-\zeta^{K_\alpha+1}},
\myendeq
which leads to 
\mybegineq
\zeta^{K_\alpha}= \frac{1-\gamma}{1-\gamma\zeta}.
\myendeq
Note that $1-\gamma\zeta=\alpha>0$. 
By taking logarithm on both sides, and substituting $\zeta$ with $\frac{1-\alpha}{\gamma}$, eventually we have 
\mybegineq
\label{eq:K_a}
K_\alpha=
\displaystyle\frac{\ln\frac{1-\gamma}{\alpha}}{\ln\frac{1-\alpha}{\gamma}}.
\myendeq
Note this result implies that under the required condition $\gamma<1$, we have a positive solution of $K_\alpha$. 
One can see that if $1-\gamma>\alpha$, then $1-\alpha>\gamma$; otherwise if $1-\gamma<\alpha$, then $1-\alpha<\gamma$.
Therefore $K_\alpha>0$. To further explore the properties of $K_\alpha>0$, we first introduce a lemma. 
\begin{lem}\label{lem:K_a}
$K_\alpha$ is monotonously increasing in terms of $\gamma$.
\end{lem}
The proof of Lemma~\ref{lem:K_a} is given in Appendix~\ref{proof:Ka}. 

The interpretation is rather straightforward -- the larger the $\gamma$ ratio is, the more frequent DPs arrive compared with EPs, hence causing higher chance of battery depletion. To maintain the depletion probability under increased $\gamma$, a larger battery capacity is therefore needed. 

From Lemma~\ref{lem:K_a}, we arrive at Theorem~\ref{theo:K_a}.
\begin{thm}\label{theo:K_a}
The battery size as required by the depletion probability is a monotonously decreasing function of the latter.
\end{thm}
\begin{proof}
Obviously, Eq.(\ref{eq:K_a}) contains some kind of symmetry between $\gamma$ and $\alpha$. 
Let $\alpha'=1-\alpha$, $\gamma'=1-\gamma$, then the function for calculating $K_\alpha$ satisfies
\begin{equation}
\begin{array}{rl}
f(\gamma, \alpha)
&= 
\displaystyle\frac{\ln(1-\gamma)-\ln\alpha}{\ln(1-\alpha)-\ln\gamma}
=\displaystyle\frac{\ln\gamma' - \ln(1-\alpha')}{\ln\alpha'-\ln(1-\gamma')} \\ 
&= \displaystyle\frac{\ln(1-\alpha') - \ln\gamma'}{\ln(1-\gamma') - \ln\alpha'} =f(\alpha',\gamma')
=f(1-\alpha, 1-\gamma).
\end{array}
\end{equation}
This suggests that for function $f(.)$, an increased $\alpha$  corresponds in effect to a decreased ``$\gamma$''; and an increased $\gamma$ corresponds to a decreased ``$\alpha$''. As we have already shown $K_\alpha$ is monotonically increasing with $\gamma$, we can now conclude
$K_\alpha$ is a monotonically decreasing function of $\alpha$. 
\end{proof}

Theorem~\ref{theo:K_a} is again reasonable, since a smaller energy depletion probability would require a larger battery size. 
Figure \ref{fig:K_a} shows some example $K_\alpha$ curves with different $\gamma$ and $\alpha$ values. 

\begin{figure}[!t]
 	\centering
 	\includegraphics[width=1.\textwidth]{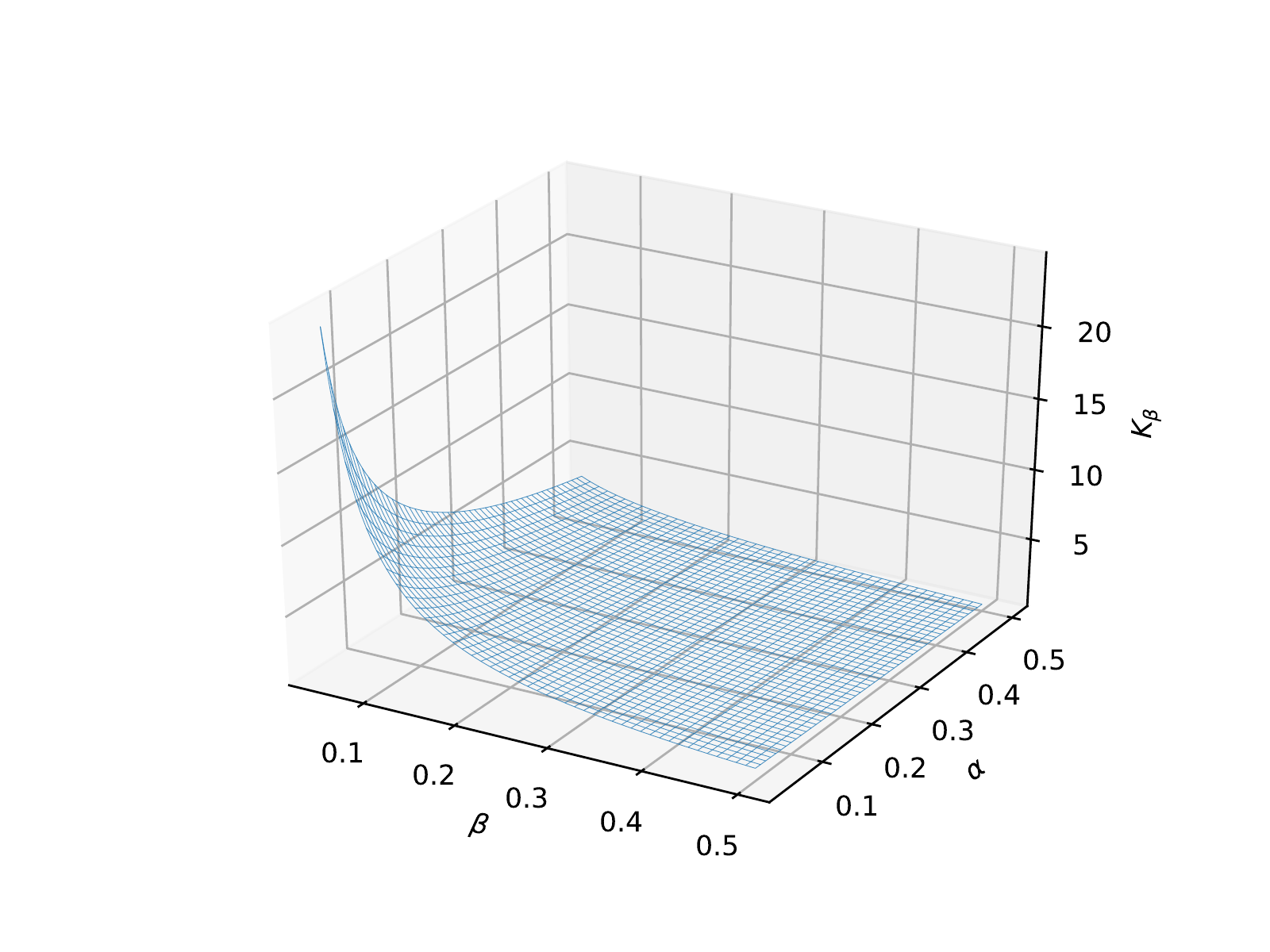}
 	\caption{$K_\beta$ values with respect to $\alpha $ and $ \beta$ when $\gamma=0.95$.}  
 	\label{fig:K_a_b}
 \end{figure}

So far we have assumed $\gamma\neq 1-\alpha$. 
The special case of $\gamma=1-\alpha$, however, 
is allowed, as again using L'H\^{o}pital's rule we have 
\mybegineq
K_\alpha=
\frac{1}{\alpha}-1,
\myendeq
which is also a positive, monotonically decreasing function of $\alpha$. 

\subsection{Battery size versus overflow probability}
Let $K_\beta$ be the battery size decided by a given $P_{E_K}$ value ($\beta$). 
We consider the normal case given in Eq.(\ref{P_K}). Let $z=1/\zeta$. We have
\mybegineq
\beta=\displaystyle\frac{1-z}{1-z^{K_\beta+1}},
\myendeq
which leads to 
\mybegineq
z^{K+1}=\frac{z+\beta-1}{\beta}.
\myendeq
Here we can see unless $z+\beta>1$, there is no real solution for $K_\beta$. In fact, it is easy to see that when $z<1$,
\mybegineq
1-z<\frac{1-z}{1-z^{K_\beta+1}}<1,
\myendeq
i.e., $\beta$ has a lower bound $1-z$. 

When the condition $z+\beta>1$ is satisfied (and naturally so when $z>1$),  we have 
\mybegineq
K_\beta=\displaystyle\frac{\ln \frac{z+\beta-1}{\beta}}{\ln z} -1=
\frac{\ln[\gamma-(1-\alpha)(1-\beta)]-\ln\beta\gamma}{\ln\gamma-\ln(1-\alpha)}
\label{eq:K_b}
\myendeq

For the battery size decided by the overflow probability, we have the following theorem:
\begin{thm}\label{theo:K_b}
$K_\beta$ is monotonically decreasing when $\beta$ increases.
\end{thm}
The proof of Theorem~\ref{theo:K_b} is given in Appendix~\ref{proof:Kb}.

Figure~\ref{fig:K_a_b} shows the trend $K_\beta$ values display across a range of $\alpha$ and $\beta$ values when $\gamma=0.95$. 
 
It can be proven that for the special case of $z=1$, i.e., $\gamma = 1-\alpha$, we have
\mybegineq
K_\beta = \lim_{z\rightarrow 1}\frac{\ln \frac{z+\beta-1}{\beta}}{\ln z} -1 = \frac{1}{\beta}-1.
\myendeq
Clearly, Theorem~\ref{theo:K_b} still holds. 

\subsection{Battery sizing algorithm}
Given different requirements in terms of $\alpha$ and $\beta$ values, and the system setup in terms of $\gamma$, we can derive
the relevant $K_\alpha$ and $K_\beta$ to size up the battery. As seen from Eq.(\ref{eq:K_b}), the condition 
\[ (1-\alpha)(1-\beta) < \gamma \] 
has to stand for calculating $K_\beta$. 
One question remains -- between $K_\alpha$ and $K_\beta$ which one actually decides the size of the battery? 
It is easy to see that when $\beta + \gamma=1$, $K_\alpha=K_\beta$. Since we have shown that $K_\alpha$ is an monotonically increasing function
of $\gamma$, and $K_\beta$ a monotonically decreasing function of $\beta$, 
we give the following corollary without needing a formal proof: 
\begin{cor}
If $\beta+\gamma<1$, $K_\alpha<K_\beta$; if $\beta+\gamma>1$, $K_\alpha>K_\beta$.
\end{cor}

\begin{figure}
\centering
\begin{subfigure}[b]{0.9\textwidth}
 \includegraphics[width=\textwidth]{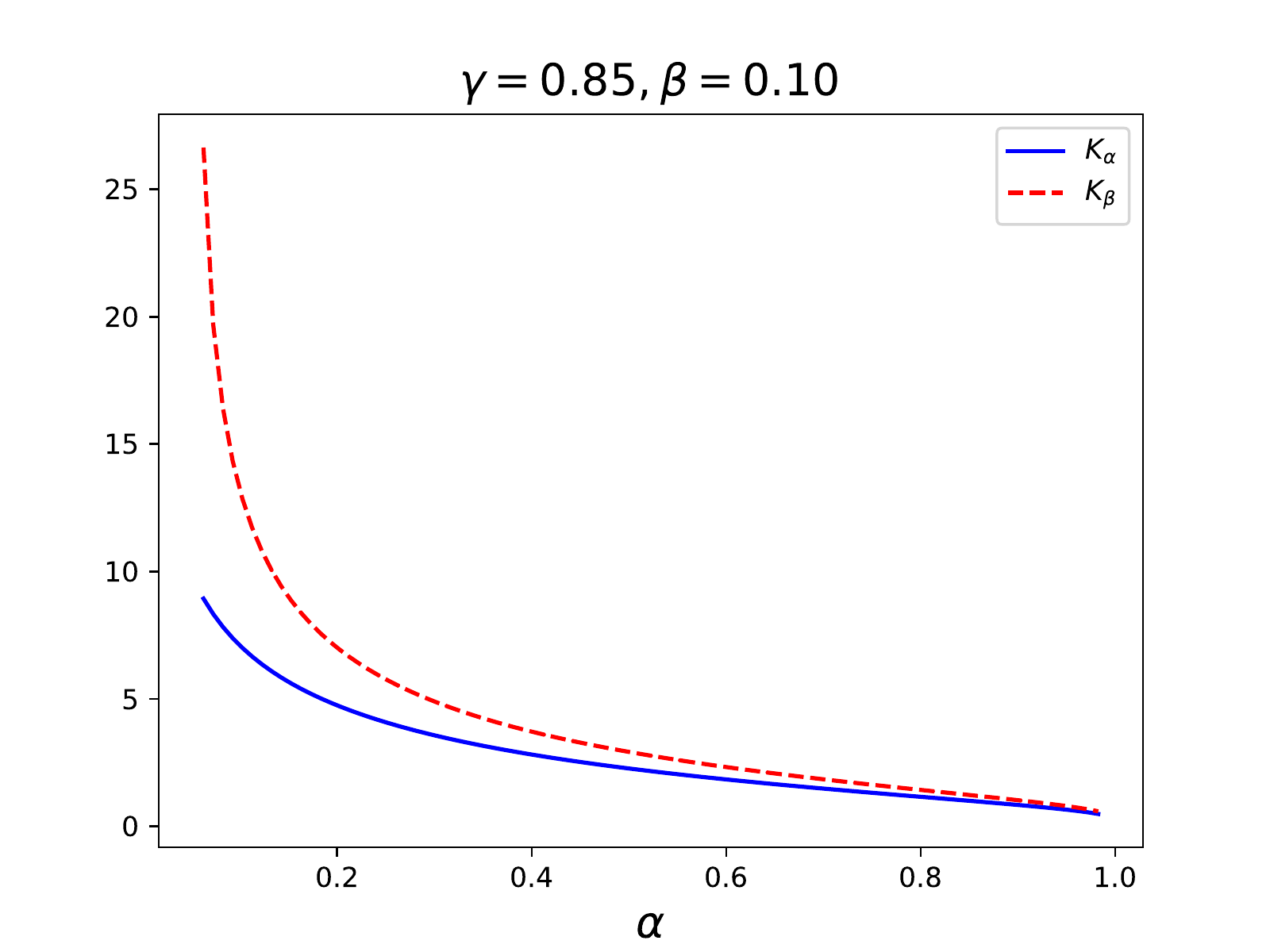}
 \caption{$\beta+\gamma<1$}
\end{subfigure}
\begin{subfigure}[b]{0.9\textwidth}
 \includegraphics[width=\textwidth]{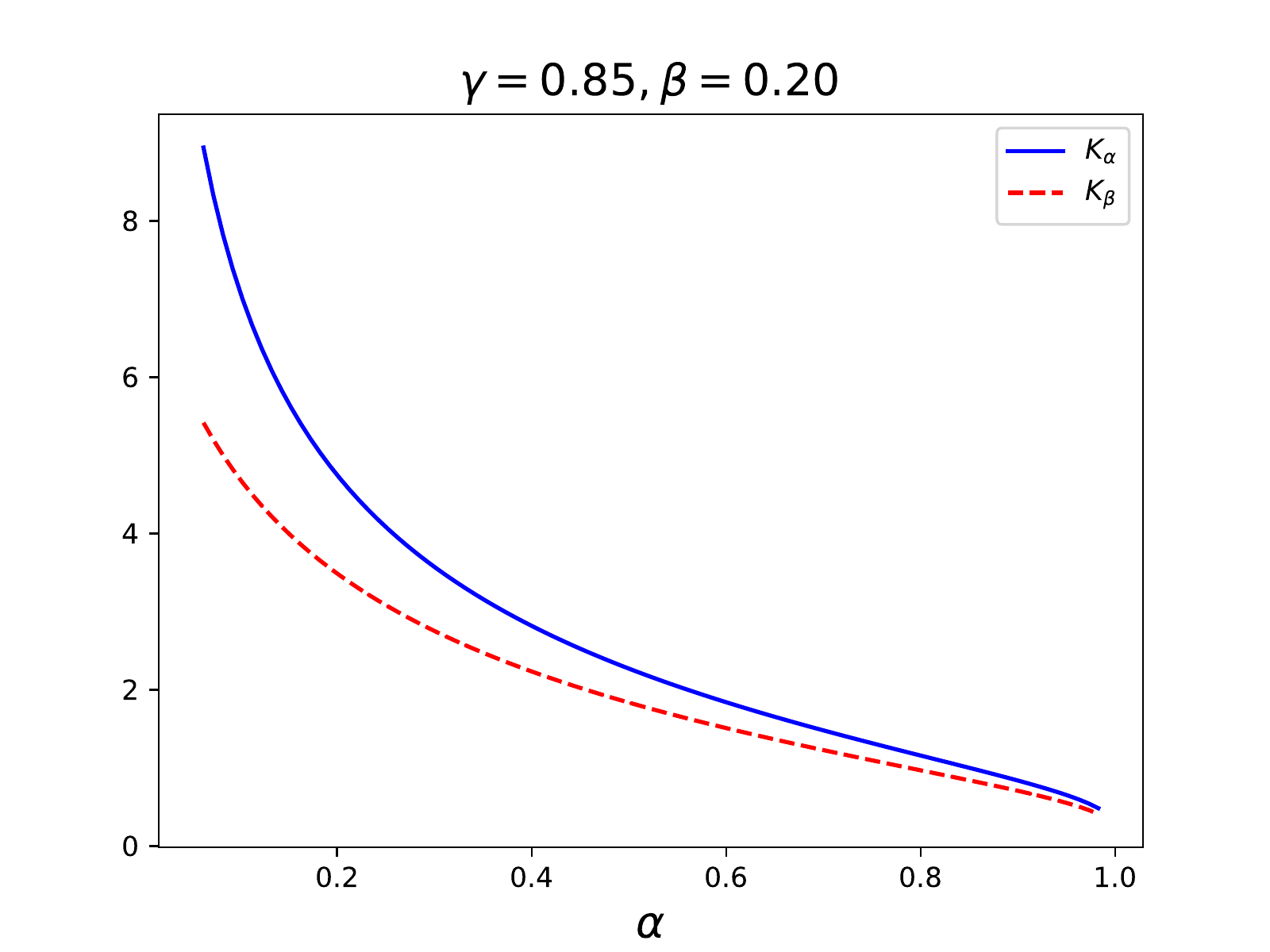}
  \caption{$\beta+\gamma>1$}
\end{subfigure}
\caption{Comparison of $K_\alpha$ and $K_\beta$ in two different cases.}
\label{fig:K_compare}
\end{figure}

Hence we have the following algorithm for battery sizing.

\begin{algorithm}[H]
\caption{Battery sizing under constraints of $P_{E_0}$ and $P_{E_k}$}
\label{algo:batt-sizing}
\KwData{$0<\alpha,\beta<1$, $(1-\alpha)(1-\beta)<\gamma<1$}
\KwResult{minimum $K$ satisfying $P_{E_0}<\alpha$, $P_{E_K}<\beta$}
 \eIf{$\beta+\gamma<1$}
 {
   $K\leftarrow K_\beta$ using Eq.(\ref{eq:K_b})\;
 }{
   $K\leftarrow K_\alpha$ using Eq.(\ref{eq:K_a})\;
 }
\end{algorithm}
 
\medskip
Figure \ref{fig:K_compare} gives two settings when $K_\alpha$ and $K_\beta$ are compared. 
These numerical examples clearly confirm the correctness of Corollary 1. 

\subsection{A numerical example}
To demonstrate the feasibility of an EH node in mobile settings, let us assume a scenario with the battery depletion probability $\alpha=0.05$, and  the overcharge
probability $\beta=0.3$. Suppose the three rates in the system are 
$\lambda_D=0.72$ data packets/sec, 
$\lambda_E=0.8$ energy packets/sec,
and $\lambda_C=0.9$ connection/sec. 
We have $\gamma=\lambda_D/\lambda_E=0.9$, and $\beta+\gamma>1$. Hence according to Algorithm 1, we obtain the required battery size
$K=K_\alpha=12.82\approx 13$ using Eq.(\ref{eq:K_a}). If we use an energy packet size of 155$\mu W$ which can be generated by a moderate walking activity~\cite{Gorlatova:15}, the needed battery size will be $13\times 155\mu W=2.015mW$.

\section{Conclusion}
Despite the great  potential in utilizing rechargeable nodes in wireless sensor networks and body area networks, the wide application of energy-harvesting IoT systems remains elusive, largely due to the uncertain nature of energy harvesting and the lack of performance analysis results for guiding system design. 
In this article, we have studied the modelling of energy charging and consumption behaviours of sensor nodes in a DTN setting, where
data transmission is subject to both the availability of sufficient energy, and the existence of sensor nodes in reachable proximity. 
A stochastic model with coupled Poisson arrival processes on the energy and data queues of the node is formulated and solved, based on which
a closed-form solution of the optimal battery size is derived to meet the specified probabilities of energy depletion and overflow.  
 
For future work, our model can be extended by considering general probability distributions for energy or data arrivals, allowing more flexible settings of data and energy packet sizes. This may enable more types of energy-harvesting sources being considered, for instance solar. Beyond using discretized energy units, continuous fluid models~\cite{Jones:11} could be employed for future investigation, where the theoretical findings may be further validated by simulation studies under realistic settings. 

\begin{acknowledgement}
I would like to dedicate this article to Emeritus Professor Martin K. Purvis, who introduced me to the wonderful world of queueing theory and encouraged me to brave the less-travelled roads in mobile ad hoc networks and IoT research. This little but meticulous work of mine certainly benefitted from the pleasant chats we had around philosophy, theology, programming languages, Donald Knuth, etc. I would also like to acknowledge that Dr. Sophie Zareei, whose PhD thesis Martin and I co-supervised, did some of the initial work on the same topic.
\end{acknowledgement}

\appendix
\section*{Appendix -- Proofs}
\subsection*{Lemma~\ref{lem:K_a}} \label{proof:Ka}
$K_\alpha$ is monotonically increasing on $\gamma$.
\begin{proof}
It can be worked out that 
\[ \frac{\partial K_\alpha}{\partial\gamma}=\displaystyle
\frac{\frac{1}{\gamma}\ln\frac{1-\gamma}{\alpha}-\frac{1}{1-\gamma}\ln\frac{1-\alpha}{\gamma}}{\ln^2\frac{1-\alpha}{\gamma}}.
\]
We want to show that this derivative is non-negative, i.e.
\[ \frac{1}{\gamma}\ln\frac{1-\gamma}{\alpha}\ge \frac{1}{1-\gamma}\ln\frac{1-\alpha}{\gamma}, \]
which is equivalent to 
\[ \left(\frac{1-\gamma}{\alpha}\right)^{\frac{1}{\gamma}} \ge \left(\frac{1-\alpha}{\gamma}\right)^{\frac{1}{1-\gamma}}. \]
Further transformation leads to 
\[ \left(1-\gamma\right)^{\frac{1}{\gamma}}\gamma^{\frac{1}{1-\gamma}} \ge
\alpha^{\frac{1}{\gamma}}(1-\alpha)^{\frac{1}{1-\gamma}}. \]
Let $x=1-\alpha$, and 
\[ f(z)=(1-z)^\frac{1}{\gamma}z^{\frac{1}{1-\gamma}}, \]
so to prove the inequality above we only need to show that $f(\gamma)\ge f(x)$. To find the maximum of $f(z)$,
let $f'(z)=0$. Solving this, we get $z=\gamma$. 
\qed
\end{proof}
The derivation above also shows that the equality stands when $\gamma=1-\alpha$. 

\subsection*{Theorem~\ref{theo:K_b}}
\label{proof:Kb}
$K_\beta$ is monotonically decreasing on $\beta$.
\begin{proof}
We consider the general case where
$K_\beta$ can be put as 
\[ K_\beta=\displaystyle\frac{\ln(\frac{z-1}{\beta}+1)}{\ln z}-1, \]
where $z=\frac{\gamma}{1-\alpha}$. 
Consider two cases only (the special case of $z=1$ is already handled in main text):
\begin{enumerate}
\item $z>1$. Both the numerator and the denominator of the fraction term are positive. Clearly the bigger $\beta$ is, the smaller $K_\beta$;
\item $z<1$. Both the numerator and the denominator are negative. With $\beta$ increasing,
$\frac{z-1}{\beta}$ will increase, albeit being negative. The numerator will increase as an negative value, hence the value for $K_\beta$ will decrease as a positive value. 
\end{enumerate}
\qed
\end{proof}

\bibliography{paper}

\end{document}